\documentclass[conference]{IEEEtran}
 \IEEEoverridecommandlockouts
\usepackage{array}
\usepackage{cite}
\usepackage{amsmath,amssymb,amsfonts}
\usepackage{graphicx}
\usepackage{textcomp}
\usepackage{xcolor}
\usepackage{amsmath, amssymb, amsthm}
\usepackage{algorithm}
\usepackage{algpseudocode}
\usepackage{url}
\usepackage{hyperref}
\usepackage{subfigure}
\usepackage{tabularx}
\usepackage{booktabs}
\usepackage{nccmath}
\usepackage{geometry}
\def\BibTeX{{\rm B\kern-.05em{\sc i\kern-.025em b}\kern-.08em
    T\kern-.1667em\lower.7ex\hbox{E}\kern-.125emX}}

\theoremstyle{plain}
\newtheorem{theorem}{Theorem}

\begin{document}
 
\title{Privacy-Utility-Fairness: A Balanced Approach to  Vehicular-Traffic Management System
}
 \author{\IEEEauthorblockN{1\textsuperscript{st} Poushali Sengupta}
 \IEEEauthorblockA{\textit{Department of Informatics} \\
 \textit{University of Oslo}\\
Oslo, Norway \\
 poushals@uio.no}
 \and
\IEEEauthorblockN{2\textsuperscript{nd} Sabita Maharjan}
\IEEEauthorblockA{\textit{Department of Informatics} \\
\textit{University of Oslo}\\
Oslo, Norway\\
sabita@ifi.uio.no}
\and
\IEEEauthorblockN{3\textsuperscript{rd} Frank Eliassen}
\IEEEauthorblockA{\textit{Department of Informatics} \\
\textit{University of Oslo}\\
Oslo, Norway \\
frank@ifi.uio.no}
\and
\IEEEauthorblockN{4\textsuperscript{th} Yan Zhang}
\IEEEauthorblockA{\textit{Department of Informatics} \\
\textit{University of Oslo}\\
Oslo, Norway \\
yanzhang@ieee.org}
 }

\maketitle
\begin{abstract}

Location-based vehicular traffic management faces significant challenges in protecting sensitive geographical data while maintaining utility for traffic management and fairness across regions. Existing state-of-the-art solutions often fail to meet the required level of protection against linkage attacks and demographic biases, leading to privacy leakage and inequity in data analysis. In this paper, we propose a novel algorithm designed to address the challenges regarding the balance of privacy, utility, and fairness in location-based vehicular traffic management systems. In this context, utility means providing reliable and meaningful traffic information, while fairness ensures that all regions and individuals are treated equitably in data use and decision-making. Employing differential privacy techniques, we enhance data security by integrating query-based data access with iterative shuffling and calibrated noise injection, ensuring that sensitive geographical data remains protected. We ensure adherence to epsilon-differential privacy standards by implementing the Laplace mechanism. We implemented our algorithm on vehicular location-based data from Norway, demonstrating its ability to maintain data utility for traffic management and urban planning while ensuring fair representation of all geographical areas without being overrepresented or underrepresented. Additionally, we have created a heatmap of Norway based on our model, illustrating the privatized and fair representation of the traffic conditions across various cities. Our algorithm provides privacy in vehicular traffic management by effectively balancing fairness and utility.
\end{abstract}


\section{Introduction}

The location-based vehicular-traffic management systems (LBVMS),\cite{moni} collects and analyzes extensive geographical data, which are inherently raising significant privacy concerns due to the detailed and sensitive nature of the information processed. The geographical data processed in LBVMS typically includes user GPS coordinates, travel patterns, origin-destination data, and time-stamped location logs. In LBVMS, such analysis of the data can reveal sensitive information such as a person’s home, workplace, or regular travel patterns. For example, in 2017, a data leak exposed the location information of approximately $800,000$ electric Volkswagen vehicles, allowing potential attackers to access precise driver movements and personal details \cite{leak}. To address these concerns, various privacy-preserving mechanisms may help to protect privacy in LBVMS. Methods like spatial cloaking techniques \cite{spcl} generalize a user's exact location into a broader area shared by multiple users, thereby providing k-anonymity \cite{k} and reducing the risk of re-identification \cite{re}. 
\par Additionally, differential privacy methods \cite{dp} introduce controlled noise into location data, ensuring that individual user information remains confidential. On the other hand, the utility of this data must not be compromised, as it is fundamental for accurate traffic forecasting, route optimization, and congestion management. In data privacy contexts, "utility" refers to the usefulness and accuracy of data after it has undergone privacy-preserving transformations, ensuring that the data remains valuable for analysis and decision-making \cite{b7}. 
 \par Moreover, ensuring fairness \cite{fair} is essential to avoid any demographic or regional bias in data analysis, which could lead to unequal resource allocation and decision-making. For example, if traffic prediction models are trained predominantly on data from urban areas, they may underperform in rural regions, resulting in inadequate infrastructure development. Similarly, models that overlook socioeconomic disparities might prioritize affluent neighbourhoods for traffic mitigation measures, neglecting less privileged communities. Addressing these fairness concerns is essential to ensure equitable transportation planning and resource distribution across all regions and demographics. Two primary types of fairness are individual fairness \cite{ifair}, which asserts that similar individuals should receive similar outcomes, ensuring equitable treatment on a case-by-case basis, and group fairness \cite{gfair}, also known as statistical or demographic parity \cite{demo}, which ensures that different demographic groups receive comparable outcomes, preventing systemic biases against any particular group. In LBVMS, group fairness can be suitable as it ensures equitable treatment across diverse regions and populations \cite{lfair}, thereby preventing biases that could lead to overrepresentation or underrepresentation in resource allocation and decision-making. Therefore, achieving an optimal balance between privacy, utility, and fairness is crucial in LBVMS \cite{b1}. 
\par Several state-of-the-art algorithms have been developed to address these challenges. However, they all have their limitations. The authors of \cite{b2} introduced a model that protects user data by adding Planar Laplace noise to GPS points, which may impact traffic monitoring accuracy. The author in \cite{b3} used a Staircase Randomized Response to boost location privacy, which, while increasing service utility, may reduce data accuracy and not fully ensure fairness. The author in \cite{b4} introduced the Rényi Fair Information Bottleneck Method (RFIB) to balance utility, fairness, and model compactness, but achieving this balance can be complex and varies with the dataset and context. The authors of \cite{b5} proposed enhancing user privacy in location-based services by eliminating centralized anonymizers and leveraging pseudonymous authentication within a decentralized framework. However, the approach may introduce complexities in managing trust among peers and ensuring system-wide security without a central authority. The author in \cite{b6} aimed to balance user location privacy with service quality by introducing a Quality of Experience (QoE) metric to quantify data utility under varying privacy levels.  However, it overlooked fairness considerations, potentially leading to unequal impacts on individual or group experiences within the system. Recent studies like \cite{fairtp} have proposed a framework that aims to improve fairness in traffic prediction, but dynamic traffic and uneven sensor distribution across regions can challenge sustained fairness, potentially reducing prediction accuracy in underrepresented areas.

\par Motivated by the challenges in existing systems, in this work, we make the following contributions:

\begin{itemize}
    \item \textbf{Balanced Privacy, Utility, and Fairness:} We propose a new framework using Differential Privacy and Query-based data access for LVBMS that balances privacy, utility, and fairness. 
    \item \textbf{Dual-Stage Shuffling and Proportional Representation:} Our algorithm introduces a dual-stage iterative shuffling process for anonymization that disrupts patterns vulnerable to linkage attacks while maintaining proportional representation across diverse regions. 
    
    \item \textbf{Experiment and Interactive Heatmap:} We validate our algorithm using Norway's vehicular-traffic data and created an interactive heatmap. This heatmap shows privatized, fair, and useful traffic predictions, offering insights for traffic management.
\end{itemize}


\section{Proposed Framework}



The proposed location-based Vehicular-traffic Management System (LBVMS) operates within an IoT-edge-cloud architecture, possessing main system parameters to perform seamless data processing while protecting user privacy. The IoT layer consists of connected devices that collect raw traffic data, including vehicle counts, speeds, and precise geolocations, which are represented as $\mathcal{D} = \{d_1, d_2, \dots, d_n\}$. These data points include sensitive attributes such as license plate numbers, vehicle identifiers, timestamps, driver characteristics, and more, referred to as Personally Identifiable Information (PII). Data anonymization techniques are applied to PII to mitigate reidentification risks before transmitting the data by generating a unique identity combining all PII, making it difficult to extract specific details while still uniquely distinguishing one vehicle from another.  Edge servers, positioned closer to data sources, perform query and data processing and aggregation, allowing for real-time analysis and reducing the need to transmit large volumes of raw data to the cloud. Queries are defined as $Q: \mathcal{D} \times \mathcal{C} \to \mathcal{R}$, where $\mathcal{C}$ specifies query constraints like various features, and $\mathcal{R}$ is the resulting dataset with required features based on the query. The framework uses a two-stage iterative shuffling mechanism at the edge to enhance privacy and prevent linkage attacks. 
\par Shuffling permutes $\mathcal{R}$ using a permutation function $\sigma$, mathematically represented as $\mathcal{R}' = \sigma(\mathcal{R})$, where $\sigma: \{1, 2, \dots, k\} \to \{1, 2, \dots, k\}$. The shuffling operation is performed iteratively for $n$ iterations to strengthen anonymisation. At each iteration $t$, the permutation function $\sigma_t$ is updated with randomness, represented as $\sigma_{t+1} = f(\sigma_t, \mathcal{P})$, where $\mathcal{P}$ is a random perturbation factor. Also, differential privacy mechanisms, such as the Laplace mechanism parameterized by a privacy budget $\epsilon$, are implemented at the edge to protect user location data further. The noisy-shuffled data set is $\mathcal{R''}$. 

\par The National Vehicular Transportation Authority (NVTA) manages the cloud layer and performs high-level analyses to identify traffic patterns and trends.  When relevant data reaches NVTA, it has been anonymized and aggregated, ensuring individual user privacy is protected. The system is designed to maintain fairness and proportional representation of subgroups $G_j$ within the dataset, ensuring group proportions $\phi_j = \frac{|G_j|}{|\mathcal{D}|}$ are preserved across processing stages. These parameters form the foundation for the privacy-preserving and efficient operation of the LBVMS, enabling secure, real-time vehicular-traffic management. Let us assume we have a query $\textbf{Q}$: "\textbf{traffic density in region $R$ at time $t$}", then \textbf{Q} will map to:
\begin{equation}
   \mathcal{R}: Q(\mathcal{D}, \mathcal{C}) = \{d_i \mid \text{location}(d_i) \in R, \text{timestamp}(d_i) = t\}.
\end{equation}
Where $\mathcal{D}$ is the user vehicular data within the region $R$ at the timeframe $t$.
The query function reduces the data surface accessible to the system, thereby minimizing the exposure of sensitive attributes irrelevant to the query.
\par Let $\mathcal{R} = \{r_1, r_2, \dots, r_k\}$ represent the result set returned by the query function $Q$. To enforce group fairness, the shuffling process is performed batch-wise, ensuring that the proportional representation of all subgroups in the dataset is maintained. The detailed Fairly Iterative Shuffling (FIS) technique is inspired by the work \cite{fairly}. The dataset $\mathcal{D}$ is divided into disjoint groups $G_j; j = 1, \dots m$, where each group represents a specific proportion $\phi_j = \frac{|G_j|}{|\mathcal{D}|}$, reflecting its size relative to the entire dataset. The shuffling process operates in two stages: first, local shuffling is performed within each group $G_j$ using a permutation function $\sigma_j$ to anonymize data points within the group. Second, a global shuffling step applies a permutation $\sigma$ across all shuffled groups $G_j'$. Local shuffling ensures anonymity within groups by randomizing the internal order of data points, while global shuffling disrupts relationships between groups, further enhancing anonymity. This two-step process preserves the size and proportional representation of each group, ensuring that $\phi_j' = \phi_j$ in the shuffled dataset $\mathcal{R}' = \sigma\left(\bigcup_{j=1}^m G_j'\right)$. Iterative repetitions of this batch-wise shuffling technique further reduce variance in group representation, converging to the original proportions in $\mathcal{D}$, thereby achieving group fairness while anonymizing data to mitigate re-identification risks. Fig. \ref{prv} demonstrates how duel iterative shuffling works in our algorithm. Each color in the diagrams corresponds to a specific area on the map, indicated by pins of the same color, linking the data explicitly to its geographic origin. Table \ref{tab} illustrates how the dataset looks before and after shuffling.
\begin{figure}
    \centering
    \includegraphics[width=\linewidth, height = 9 cm]{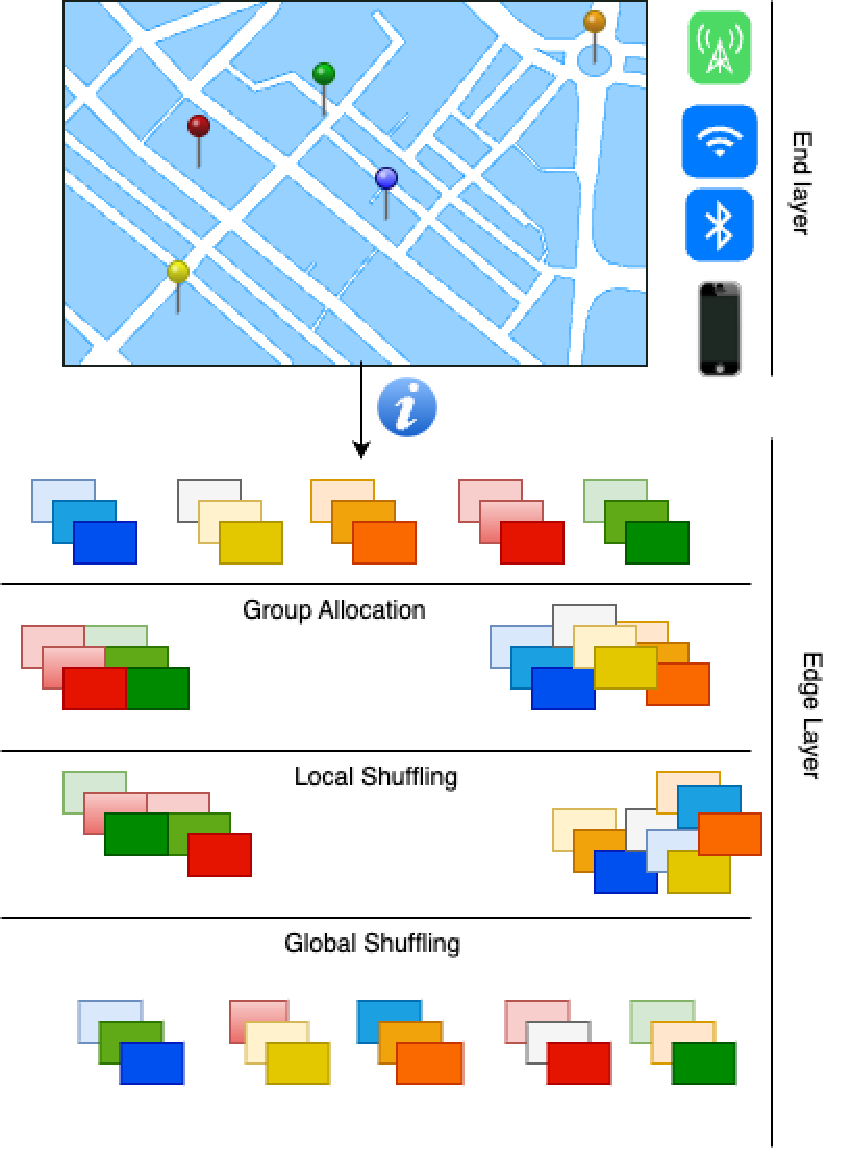}
    \caption{ Local and Global Shuffling}
    \label{prv}
\end{figure}

\begin{table}[ht]
\small
\centering
\caption{Comparison of Original and Shuffled Datasets}
\label{tab:dataset_comparison}
\begin{tabular}{>{\raggedright\arraybackslash}p{0.45\columnwidth} >{\raggedright\arraybackslash}p{0.45\columnwidth}}
\toprule
\textbf{Original Dataset} & \textbf{Shuffled Dataset} \\
\midrule
\textbf{Anonymized Vehicle ID, Location, Speed(km/h), Count} & \textbf{Anonymized Vehicle ID, Location, Speed (km/h), Count} \\
A1B2C3, (59.91, 10.75), 60, 5 & D4E5F6, (59.91, 10.75), 55, 3 \\
D4E5F6, (59.92, 10.76), 55, 3 & N4N5O1, (59.90, 10.74), 60, 6 \\
G7H8I9, (59.91, 10.75), 50, 4 & A1B2C3, (59.92, 10.76), 60, 5 \\
J1K2L3, (59.90, 10.74), 65, 2 & G7H8I9, (59.92, 10.76), 50, 4 \\
N4N5O1, (59.92, 10.76), 60, 6 & J1K2L3, (59.91, 10.75), 65, 2 \\
\bottomrule
\end{tabular}
\label{tab}
\end{table}

\par After the data has been shuffled, the framework ensures differential privacy by injecting carefully noise into the data using the Laplace mechanism, which ensures that the inclusion or exclusion of an individual's data at any single time point does not significantly affect the overall analysis, thereby protecting individual privacy across time \cite{dp}. Noise is injected using the Laplace mechanism in  $\mathcal{R}'$ and the noisy dataset will be, $\mathcal{R}'' = \mathcal{R}' + \mathcal{N}$, where, $\mathcal{N} \sim \text{Laplace}(0, b),$ and $b$ is the scale parameter defined as $b = \frac{\Delta f}{\epsilon}$. Here, $\Delta f$ is the sensitivity of the query function $Q$, defined as, $\Delta f = \max_{\mathcal{D}, \mathcal{D}'} \|Q(\mathcal{D}) - Q(\mathcal{D}')\|_1$ where, $\mathcal{D}$ and $\mathcal{D}'$ differ by at most one record and  $\epsilon$ is the privacy budget. The noise injection process is governed by $\epsilon$, which controls the trade-off between privacy and utility through  Differential Privacy, which ensures that the probability of a query result from any dataset differs from a similar dataset by at most $e^\epsilon$ \cite{dp}. A smaller $\epsilon$ enhances privacy at the utility cost, while a larger $\epsilon$ improves utility. The framework implements a privacy budget allocation strategy \cite{dp}, where the total privacy budget $\epsilon_T$ is distributed across multiple queries. Subsequent queries consume progressively smaller portions of the budget, ensuring cumulative privacy leakage remains within acceptable limits. 
\par To ensure longitudinal privacy, the total privacy budget $\epsilon_T$ is distributed across multiple queries $Q_1, Q_2, \dots, Q_m$ using a decaying allocation strategy, $\epsilon_T = \sum_{i=1}^m \epsilon_i$; where $\epsilon_i$ decreases with each subsequent query, ensuring stronger privacy guarantees over time. To optimize the performance of the system, the framework adopts an empirical risk minimization strategy to minimize the overall loss function $\mathcal{L}(U, P)$,
\begin{equation}
   \mathcal{R}isk(U, P)=\arg \min_{\epsilon, \sigma} \mathcal{L}(U, P) = \alpha \cdot (1 - U) + \beta \cdot P,
\end{equation}
where $\mathcal{L}$ is the total loss, $U$ represents the utility of the data and $P$ represents the privacy loss. $\alpha$ and $\beta$ are trade-off parameters.
\begin{theorem}
    Iterative shuffling ensures that the proportional representation of each group $ G_j$ in the dataset $\mathcal{D}$ is preserved in the shuffled dataset $\mathcal{R}'$ with high probability as the number of iterations $n \to \infty$.
\end{theorem}
 
\begin{proof}
   We have the original dataset $\mathcal{D}$, divided into $m$ groups $G_1, G_2, \dots, G_m$ , with each group $G_j$ containing $|G_j|$ elements. The total dataset size is $|\mathcal{D}|$.  After a single iteration, the shuffled dataset is, $\mathcal{R}' = \sigma\left(\bigcup_{j=1}^m G_j'\right)$. The proportion of group $G_j$ in the shuffled dataset after one iteration can be expressed as: $\phi_j^{(1)} = \frac{|G_j^{'(1)}|}{|\mathcal{R}'|}$, where $ |G_j^{'(1)}| $ represents the number of elements from group $G_j$ after local shuffling in the shuffled dataset $\mathcal{R}'$. As the local shuffling happens within internal elements, it does not change the number of elements in the group. After $n$ iterations of global shuffling, let $\phi_j^{(n)}$ represent the proportion of group $G_j$ in the shuffled dataset. Since the shuffling is random, the expected value of $\phi_j^{'(n)}$ remains equal to the original proportion $\phi_j$:
\begin{equation}
    \mathbb{E}[\phi_j^{'(n)}] = \phi_j
\end{equation}
To quantify the variability in group representation after $n$ iterations, we calculate the variance:
\begin{equation}
    \text{Var}(\phi_j^{'(n)}) = \frac{1}{n} \sum_{t=1}^n (\phi_j^{'(t)} - \phi_j)^2.
\end{equation}
The variance decreases with the number of iterations $n$, as randomness spreads the representation evenly across the shuffled dataset. Specifically:
$\text{Var}(\phi_j^{'(n)}) \propto \frac{1}{n}$. As $n \to \infty$, the variance approaches zero, $\text{Var}(\phi_j^{'(n)}) \to 0.$. This implies that the proportions $\phi_j^{'(n)}$ converge to the original proportions $\phi_j$: $\phi_j^{'(n)} \to \phi_j \quad \text{as} \quad n \to \infty.$  

\end{proof}
\begin{theorem}
The Laplace mechanism, combined with iterative shuffling, ensures that the error introduced by noise injection is uniformly distributed across all subgroups $G_j$ in the dataset, maintaining fairness through proportional representation in the final output dataset  $\mathcal{R}''$.
\end{theorem}

\begin{proof}
    
 The shuffled dataset $\mathcal{R}'$ is obtained after applying local and global shuffling to the result set $\mathcal{R} = \{r_1, r_2, \dots, r_k\}$. The final noisy dataset $\mathcal{R}''$ is defined as, $\mathcal{R}'' = \mathcal{R}' + \mathcal{N}$, where $\mathcal{N} \sim \text{Laplace}(0, b)$ and $b = \frac{\Delta f}{\epsilon}$.
In local shuffling, the data points within each group $G_j$ are permuted. After local shuffling, the probability of any data point $r_i \in G_j$ being selected is:
\begin{equation}
    \Pr[r_i \in G_j \mid \mathcal{R}'] = \frac{1}{|G_j|},
\end{equation}
where $|G_j|$ is the size of group $G_j$. In global shuffling, all data points in $\mathcal{R}$ are randomly permuted across the entire dataset. This process preserves the overall proportional representation $\phi_j$, defined as, $\phi_j = \frac{|G_j|}{|\mathcal{R}|}$. After global shuffling, the probability of any data point $r_i \in G_j$ being selected is:
\begin{equation}
    \Pr[r_i \in G_j \mid \mathcal{R}'] = \phi_j.
\end{equation}
Thus, local and global shuffling together ensure proportional representation of groups in $\mathcal{R}'$.
The Laplace noise $\mathcal{N}$ is applied independently to each data point in $\mathcal{R}'$. The noise follows the distribution $\Pr[\mathcal{N} = x] = \frac{1}{2b} \exp\left(-\frac{|x|}{b}\right)$, where the expected noise magnitude and variance are $\mathbb{E}[\mathcal{N}] = 0$ and  $\text{Var}(\mathcal{N}) = b^2$. Since $\mathcal{N}$ is independent of group membership, the expected noise magnitude and variance are identical for all groups $G_j$. Since Laplace noise is unbiased and independent of group membership, it does not alter the proportional representation of groups in $\mathcal{R}''$:
\begin{equation}
    \Pr[r_i \in G_j \mid \mathcal{R}'' ] = \phi_j.
\end{equation}
The variance of the noisy data for group $G_j$ is, $\text{Var}(\mathcal{R}'' \mid G_j) = \text{Var}(\mathcal{R}' \mid G_j) + \text{Var}(\mathcal{N} \mid G_j)$. Using the properties of Laplace noise, $\text{Var}(\mathcal{N} \mid G_j) = b^2$. Since $\text{Var}(\mathcal{R}' \mid G_j)$ is determined by shuffling and is independent of $j$, the total variance for $\mathcal{R}''$ is uniform across all groups.
\end{proof}

\section{experiment}

\subsection{Data Availablity}
In this experiment, we generate a traffic dataset for 50 Norwegian cities and towns, capturing hourly traffic density over 24 hours. The data incorporates essential factors such as weather, day of the week, and holidays. Weekday peak hours (5:30 AM to 8 AM and 4 PM to 6 PM) exhibit higher traffic, particularly in larger cities, while weekend traffic decreases on Sundays and slightly increases on Saturdays. To make the simulation more realistic, we also included weather data and temperature variations ranging from -5°C to 25°C. Rainy and foggy conditions are modelled to increase traffic density by 10\%, reflecting slower driving speeds and occasional accidents, which could add up to 100 extra vehicles per hour \cite{speed}. Snowy conditions lead to a 20\% increase, representing greater driving challenges \cite{temp}. Conversely, clear weather results in a modest rise of 20 vehicles per hour, reflecting increased travel activity under favourable conditions \cite{peng}. The dataset and code are openly available as historical\textunderscore traffic\textunderscore data.csv  and location.py respectively at in \url{https://github.com/Psxxg/Fairly-Private}.

\subsection{Result}
We developed and implemented a comprehensive algorithmic framework to simulate and analyze traffic data, focusing on privacy-preserving techniques. It retrieves data by region and time, applies global and local shuffling for anonymity, and adds Laplace noise for differential privacy. Fig. \ref{opti} shows the trade-off between privacy and data utility in traffic density analysis using differential privacy. As the privacy budget, epsilon $\epsilon$, increases, the Mean Squared Error (MSE) and Mean Absolute Error (MAE) decrease, indicating less data distortion. A vertical dashed line marks the optimal $\epsilon$ that balances privacy with data accuracy, essential for effective traffic management while maintaining privacy. This point suggests that slight increases in $\epsilon$ can significantly enhance data utility without compromising privacy. By testing different privacy budgets ($\epsilon$), we found $\epsilon = 2$ minimizes errors like mean squared error (MSE) while preserving utility. This noise level is then applied to the dataset. Figure \ref{traf} illustrates the application of Iterative Shuffling with differential privacy to vehicular-traffic scenarios in Norwegian regions, ensuring fairness by preventing overrepresentation or underrepresentation. The comparison of original and noise-injected traffic densities shows high consistency, highlighting uniform noise distribution across regions.
\begin{figure}
    \centering
    \includegraphics[width=\linewidth, height = 5cm ]{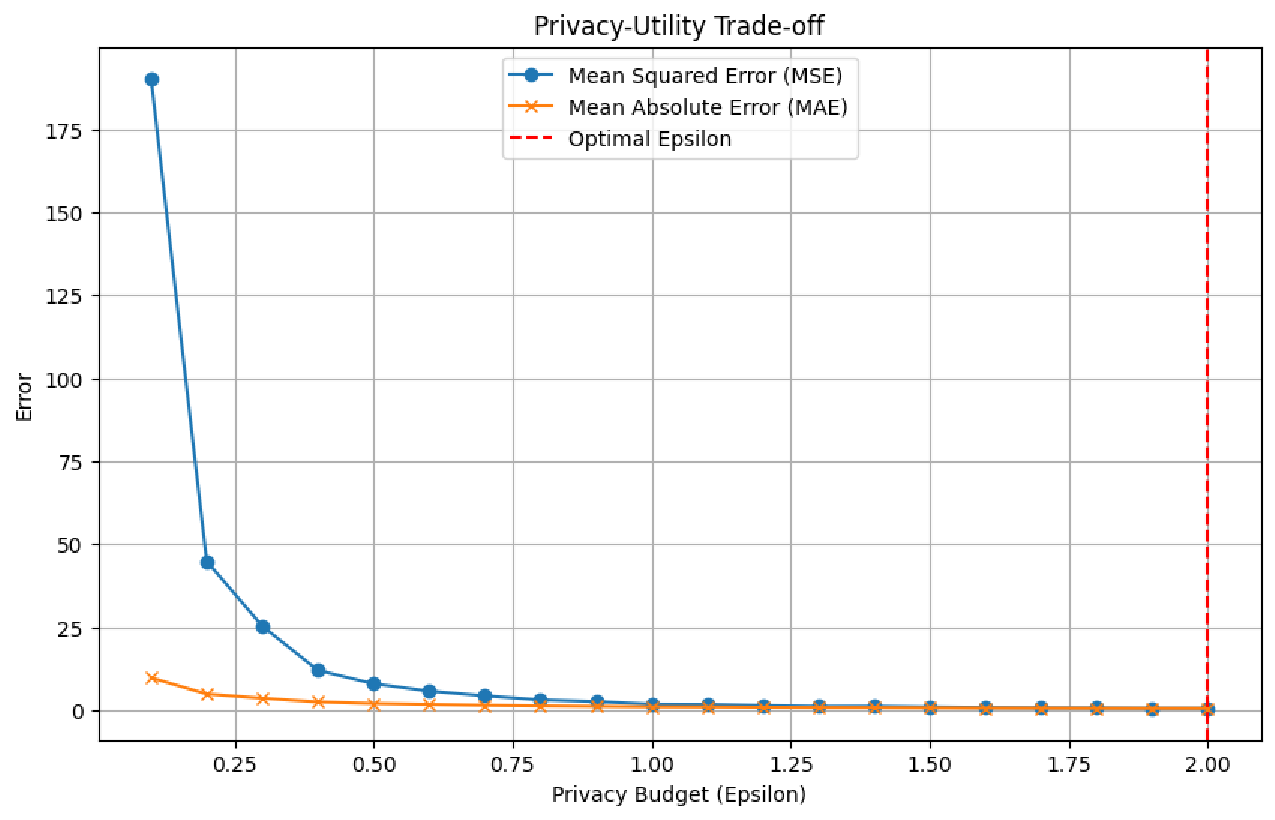}
    \caption{Privacy vs Utility}
    \label{opti}
\end{figure}
\begin{figure*}
    \centering
    \includegraphics[width=\linewidth, height = 6cm]{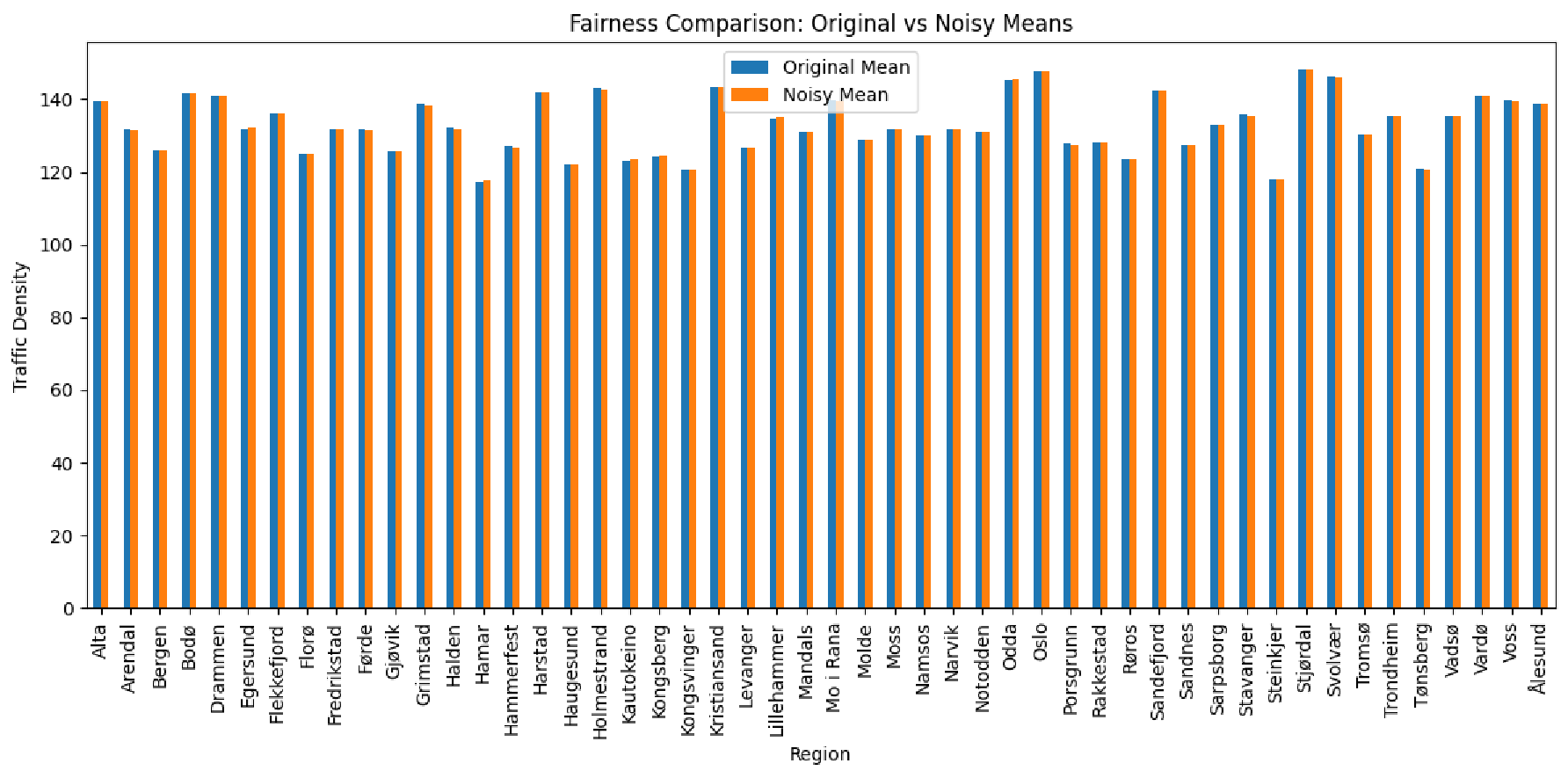}
    \caption{Uniform Application of Noise: Demonstrating unbiased noise distribution across Norwegian regions, ensuring no overrepresentation or underrepresentation.}
    \label{traf}
\end{figure*}

\par Figure \ref{heat} presents heatmaps of vehicular-traffic density at 17:00 and 7:00, scaled from 0 (no traffic) to 1 (maximum traffic). The 17:00 heatmap shows higher traffic density in several regions, indicating peak conditions as commuters return home. In contrast, the 7:00 heatmap predominantly displays lower traffic levels, with most regions showing minimal activity, characteristic of early morning hours before the rush begins. These visualizations illustrate how traffic density patterns vary temporally across Norwegian regions, highlighting the system's ability to capture such dynamics while maintaining privacy and fairness. Fig. \ref{heat} shows a balanced trade-off between privacy, utility, and fairness in traffic data across regions, with noise applied proportionally to maintain fairness. Even with Iterative Shuffling and noise, traffic patterns remain clear, helping identify congested areas without revealing individual details. 
\begin{figure*}
    \centering
    \begin{subfigure}{}
     \includegraphics[width=0.48\textwidth, height = 5cm]{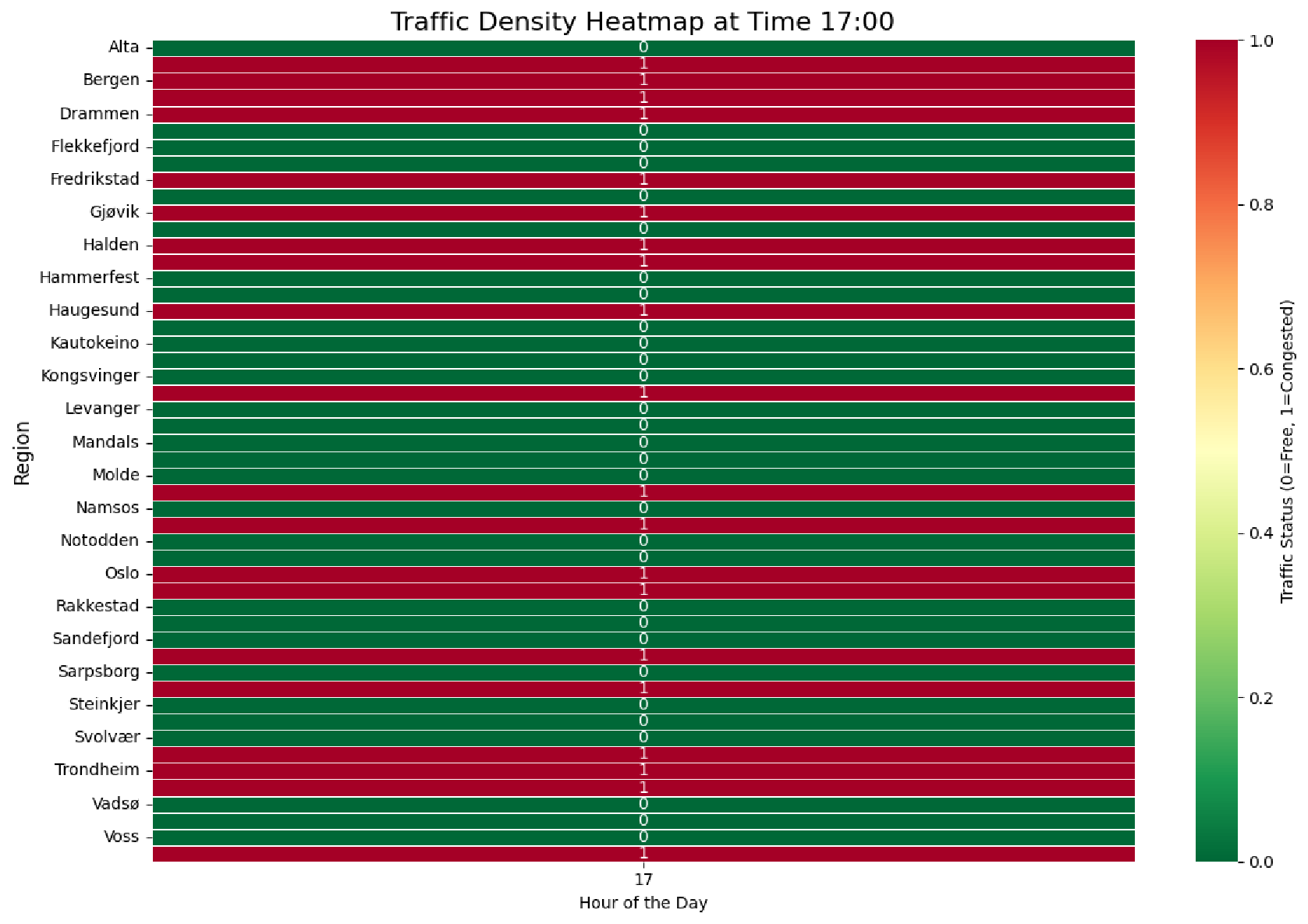}
    \end{subfigure}
    \label{heat1}
    \begin{subfigure}{}
    \centering
    \includegraphics[width=0.48\textwidth, height = 5 cm]{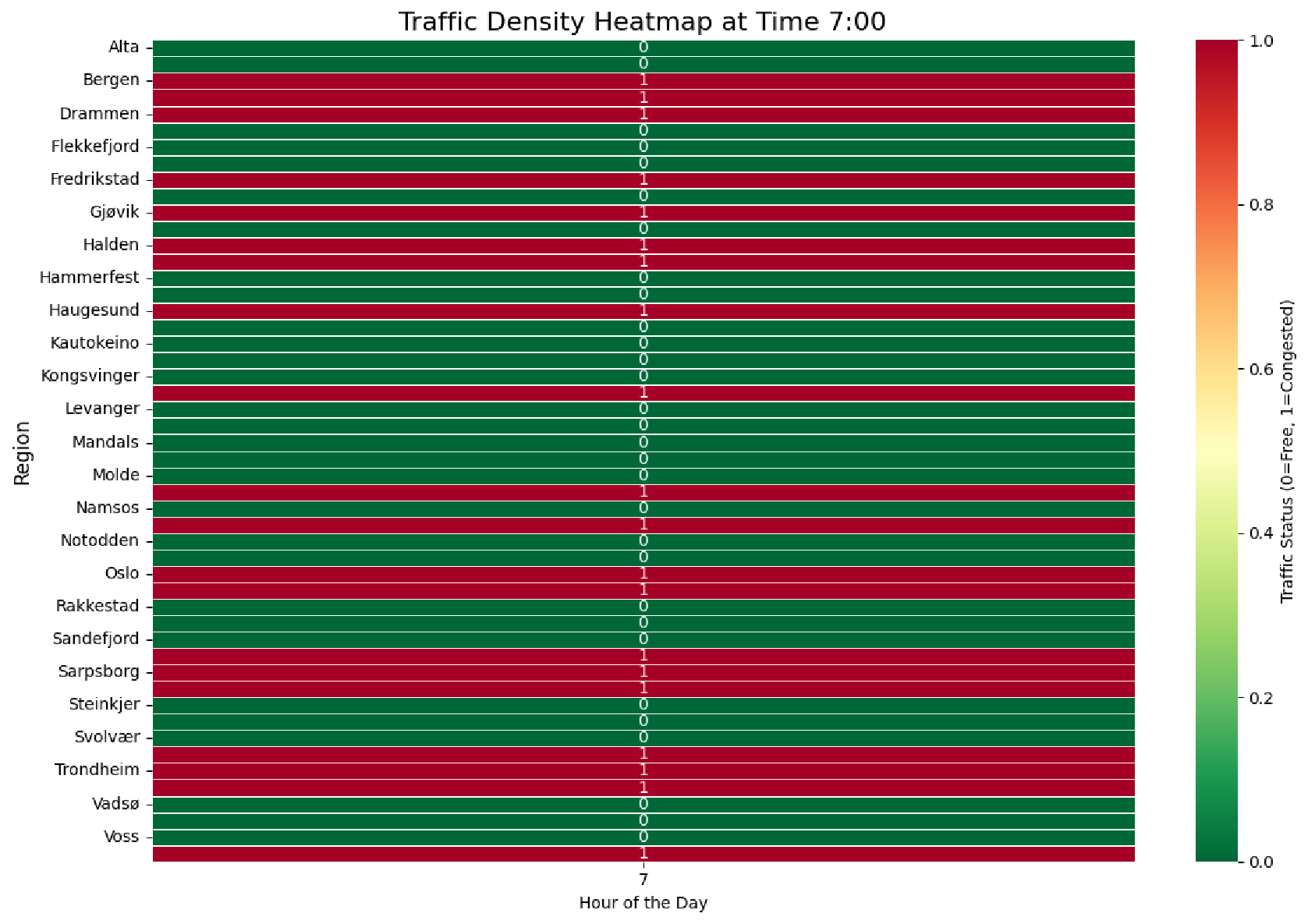}
    \label{heat2}
\end{subfigure}
\centering
\caption{Traffic congestion heatmap for Norway, showing conditions at 17:00 and 7:00 across different regions.  }
    \label{heat}
\end{figure*}

Fig. \ref{pred} shows traffic predictions across different regions for the next 24 hours, comparing results from the original data with those modified by differential privacy noise. The predictions from the original data capture typical daily traffic patterns, such as peak and off-peak hours, providing high precision that is useful for traffic management. However, the predictions with added noise, represented by dashed lines, introduce some variability and imprecision to obscure exact traffic volumes. This helps protect privacy by preventing the identification of detailed regional behaviours. Additionally, the figure highlights how combining noisy data with Iterative Shuffling improves fairness by ensuring that all regions are treated equally, with no region disproportionately represented or discriminated against due to precise traffic patterns.
\begin{figure}
    \centering
    \includegraphics[width=\linewidth ]{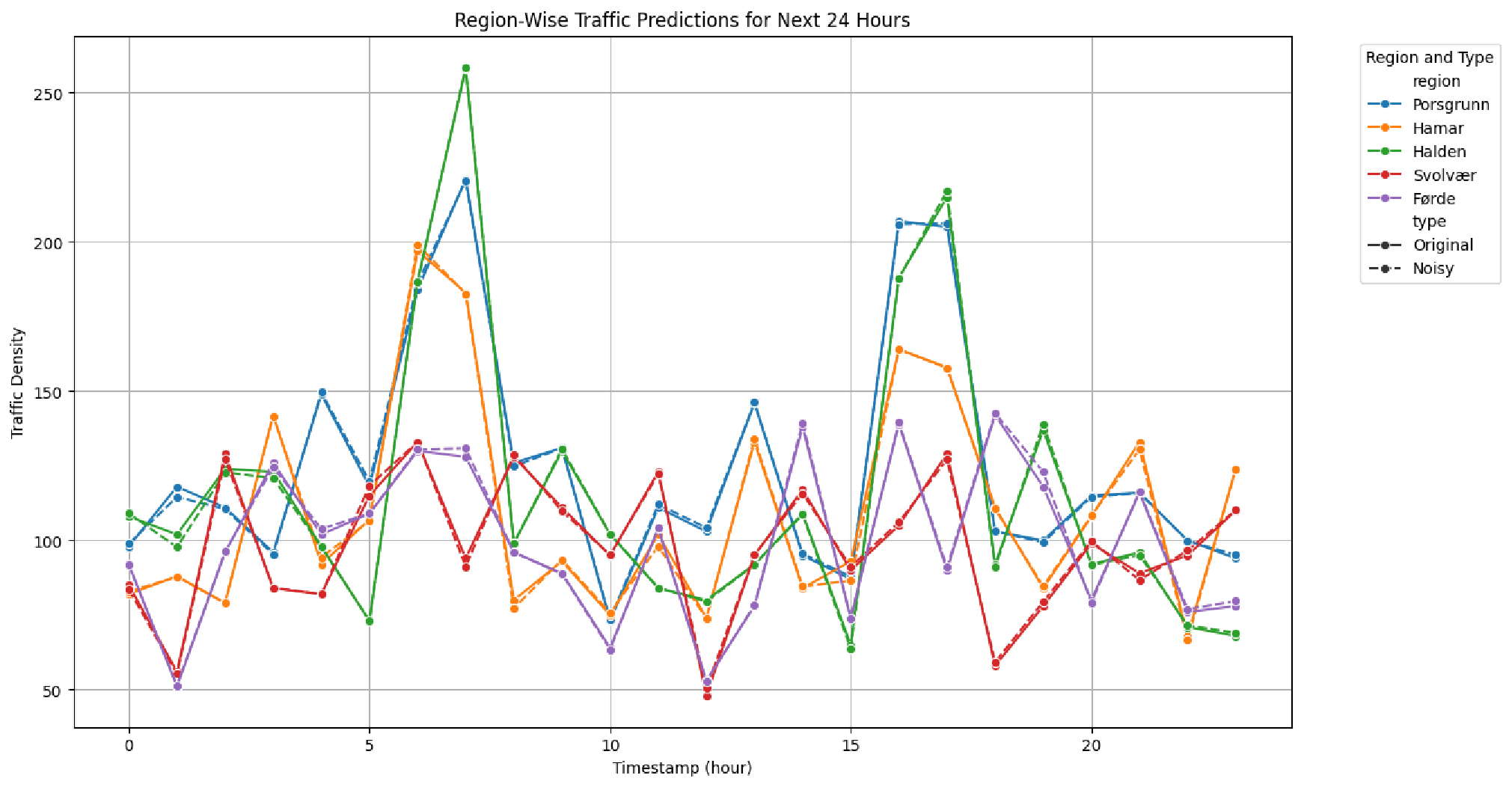}
    \caption{Region-Wise 24-Hour Traffic Predictions: Original vs. Noisy densities across Norwegian regions.}
    \label{pred}
\end{figure}
\subsection{Interactive Heatmap}
We develop an interactive heatmap using our proposed algorithm, showcasing vehicular-traffic conditions across Norway and its major cities at two distinct times. This heatmap demonstrates a well-calibrated balance between privacy, utility, and fairness, ensuring data remains useful while protecting individual information and treating all regions equitably. The heatmap shows vehicular-traffic intensity using colour gradations—darker colours represent higher traffic densities, while lighter colours indicate lower densities. It helps predict and identify traffic congestion, patterns, accident-prone areas, and road blockages at specific times without exposing exact traffic counts by displaying a noisy number of vehicles, ensuring sensitive information remains protected. The interactive heatmap also considers important factors like weather conditions, weekdays, and time of day to provide a realistic and accurate view of vehicular-traffic dynamics.  Iterative shuffling and differential privacy ensure that individual data remains protected while keeping the data useful for practical analysis. The heatmap’s interactivity enhances its utility, allowing users to explore vehicular-traffic conditions across various regions and time frames. The open street heatmap is available to download from \url{https://github.com/Psxxg/Fairly-Private/blob/main/traffic_prediction_with_graphs.html}
\section{Acknowledgment}
The Research Council of Norway has generously supported this research under the Transport 2025 Project CRISP (CRItical SPeed function and Automatic speed control function ahead of a dangerous road section), Project Number: 302327. The authors extend their profound gratitude for the funding and support, which have been pivotal in pursuing innovative solutions to improve road safety. 
\section{Conclusion and Future Direction}
Our proposed algorithm effectively balances privacy, utility, and fairness in vehicle management systems using differential privacy techniques. By carefully adding noise and using a two-step shuffling process, traffic data is useful for analysis while ensuring all areas are fairly represented. This method transforms how we manage traffic data privacy and could be used in different urban settings. We could improve the system by processing data in real-time to respond faster and work on a larger scale, potentially applying these techniques to other areas that need strong privacy protection.

\end{document}